\documentclass[conference,letterpaper]{IEEEtran}

\IEEEoverridecommandlockouts

\usepackage[utf8]{inputenc} 
\usepackage[T1]{fontenc}
\usepackage{url}
\usepackage{ifthen}
\usepackage{cite}
 \usepackage[cmex10]{amsmath} 
 \usepackage{mathrsfs}
\usepackage{amssymb,bm, bbm}

 \usepackage{tikz}
 \usepackage{graphicx}
 \usepackage{pgfplots}
\usepackage{subcaption}
\usepackage{stfloats}

\usepackage[
  letterpaper,
  left=0.64in,
  right=0.64in,
  top=0.76in,
  bottom=1.1in
]{geometry}

\newcommand{\rmd}{\mathrm{d}}
\newcommand{\bbE}{\mathbb{E}}\newcommand{\rme}{\mathrm{e}}

\newcommand{\bbR}{\mathbb{R}}

\newcommand{\cX}{\mathcal{X}}

\newtheorem{thm}{Theorem}

\newtheorem{prop}{Proposition}
\newtheorem{cor}{Corollary}

\newtheorem{rem}{Remark}

\newtheorem{assumption}{Assumption}

\title{ Sub-Gaussian Concentration and Entropic Normality of the Maximum Likelihood Estimator }
\author{%
   \IEEEauthorblockN{Leighton P. Barnes \IEEEauthorrefmark{1},
                     Alex Dytso\IEEEauthorrefmark{2}                     }
\IEEEauthorblockA{\IEEEauthorrefmark{1}%
                     Carnegie Mellon University,
                     Pittsburgh, PA, USA,
                     leightonb@cmu.edu}
   \IEEEauthorblockA{\IEEEauthorrefmark{2}%
Qualcomm Flarion Technology, Inc., Bridgewater, NJ 08807, USA,
                     odytso2@gmail.com}
 }

\begin{document}

\maketitle

\begin{abstract}
 It is well known that, under standard regularity conditions, the maximum likelihood estimator (MLE) satisfies a central limit theorem and converges in distribution to a Gaussian random variable as the sample size grows. This paper strengthens this classical result by developing several stronger
forms of asymptotic normality for the normalized MLE. With additional assumptions on the score, we first establish sub-Gaussian tail bounds and convergence of all moments for the normalized estimation error. We then prove an entropic central limit theorem for a smoothed version of the estimator, showing convergence in relative entropy to the limiting Gaussian law. When the Fisher information of the normalized estimate is bounded, or its density has bounded first derivative, we further show that the smoothing can be removed, yielding entropic normality of the MLE itself. The proofs develop auxiliary tools that may be of independent interest, including exponential consistency bounds, high-moment estimates, and entropy-control arguments for the estimator.
\end{abstract}

\section{Introduction}
Suppose we have samples $X_1,\ldots,X_n$ drawn i.i.d. from a density $f(x|\theta)$ with a parameter $\theta\in\Theta\subset \mathbb{R}$. The density $f(x|\theta)$ is with respect to a base measure $\mu$ on the sample space $\mathcal{X}$. This paper investigates the asymptotic behavior of the distribution of the maximum likelihood estimator (MLE). In particular, we study the distribution of the properly normalized estimation error $Z_n =\sqrt{n}\,(\widehat{\theta}_n - \theta_0)$ as the sample size $n \to \infty$. Under standard regularity conditions, this quantity is known to satisfy a classical central limit theorem (CLT) and converge in distribution to a Gaussian random variable. 


The main goal of this paper is to strengthen this classical result in several directions. These strengthened results come at the expense of additional assumptions on the score, which can be considered natural in many statistical models. 
First, we establish sub-Gaussian tail behavior for the normalized estimation error $Z_n$,
and use this control to upgrade convergence in distribution to convergence of moments.
Second, we prove an entropic form of asymptotic normality for a smoothed version of the estimator, $\widetilde Z_n=\sqrt{1-\varepsilon}Z_n+\sqrt{\varepsilon}Z,$
where $Z\sim\mathcal N(0,1)$ is independent of $Z_n$. Finally, under an additional Fisher information (or bounded first derivative) regularity condition on the densities of $Z_n$, we remove the smoothing and obtain entropic normality of $Z_n$ itself. Convergence in relative entropy is significantly stronger than convergence in distribution and belongs to a class of strong statistical divergences. In particular, convergence in relative entropy implies convergence in total variation distance via Pinsker's inequality and convergence in Wasserstein distance via Talagrand's inequality.

\subsection{Prior Work}

The asymptotic normality of the MLE is a classical result in statistics and forms one of the cornerstones of parametric inference. Under suitable regularity conditions, it is well known that the properly normalized estimation error $\sqrt{n}(\widehat{\theta}_n-\theta_0)$ converges in distribution to a Gaussian random variable whose variance is determined by the Fisher information \cite{cramer1999mathematical,wald1943tests,lecam1970assumptions}.  This CLT result can be extended to a large class of estimators known as M-estimator and L-estimators; the interested reader is referred to classical text such as \cite{van2000asymptotic,lehmann1998theory}.

 The entropic CLT for the sample mean was first studied by \cite{linnik1959information} and later considerably generalized by \cite{barron1986entropy}. There has also been recent activity around finding the rates of convergence and the interested reader is referred to~\cite{artstein2004rate,johnson2004fisher,madiman2007generalized,bobkov2013rate,courtade2018quantitative,bobkov2019renyi} and references therein.  
Besides the sample mean, the entropic CLT is also known to hold for all central order statistics \cite{cardone2022entropic}.

Recent work in learning theory has studied the Kullback--Leibler divergence of MLEs, especially monotonicity properties \cite{viering2019open,sellke2025learning}. However, those works consider KL divergence over the sample space between the distributions induced by the estimated and true parameters. This differs from our goal, which is to establish a stronger mode of convergence for the distribution of the MLE itself.

\subsection{Outline}
Section~\ref{sec:preliminaries} introduces the necessary preliminaries, formulates the problem, and presents several illustrative examples. Section~\ref{sec:main} states our main result. Section~\ref{sec:proofs} is dedicated to proofs. In Section \ref{sec:conclusion} we conclude and discuss some future directions for this work. The remainder of this section establishes the notation used throughout the paper. 

Random variables are denoted by upper case letters (e.g., $X$), while specific deterministic values they can take are denoted by lower case letters (e.g., $x$). Throughout the paper, logarithms are taken in the natural base.

For two random variables $W \sim f_W$ and $V \sim f_V$, the relative entropy (KL divergence) is defined as
\[
D(W\|V) = \int_{\mathbb{R}} f_W(x)\log \frac{f_W(x)}{f_V(x)}\, dx.
\]
We denote by $\mathcal{N}(0,1)$ the standard normal distribution and use $Z$ to denote a standard normal random variable. Convergence in distribution is denoted by $\xrightarrow{d}$.
 We use the standard asymptotic notation: for two functions $f$ and $g$, we write $f(n) = O(g(n))$ if there exists a constant $C>0$ and $n_0$ such that $|f(n)| \le C g(n)$ for all $n \ge n_0$. 

\section{Problem Definition and Some Preliminaries} 
\label{sec:preliminaries}
In this section we set up the required technical preliminaries. We will make a number of definitions and assumptions here that are standard in the existing literature/results regarding the asymptotic normality of the MLE. The more nonstandard assumptions that we make regarding the tail behavior of the score and the smoothness of the statistical model we enumerate explicitly in Section \ref{subsec:additional-regularity}.

We assume the density $f(\cdot|\theta)$ has common support at each $\theta$ (or equivalently we just assume $f(x|\theta)>0$, since we can ignore parts of the sample space with zero probability density). We consider a parameter space $\Theta$ that is an open subset of $\mathbb{R}$ and that has a compact closure $K$ with $|\theta|<B$ for all $\theta\in K$. Let
\begin{align}
    L_n(\theta) = \frac{1}{n}\sum_{i=1}^n  \log f(X_i|\theta)
\end{align}
be the normalized log-likelihood. Let $\widehat{\theta}_n(X_1,\ldots,X_n)$ be a maximum likelihood estimate of $\theta$ from the data $X_1,\ldots,X_n$ in the sense that
\begin{equation}
    \widehat{\theta}_n \in \arg \max_{\theta\in K}L_n(\theta)  . \label{eq:def_ml_estimation}
\end{equation}
We assume that $\log f(x|\theta)$ is twice-differentiable as a function of $\theta$ on the set $\Theta$, and that it is continuous on the closure $K$ (for almost all $x\in\mathcal{X}$). By this continuity, $L_n(\theta)$ achieves its maximum over $K$. Note that if $\widehat{\theta}_n$ is in $\Theta$ itself instead of just its closure, then it satisfies
\begin{align} L'_n(\widehat{\theta}_n) = 0 \; . \end{align}

The key quantity that we seek to study is the normalized error given by
\begin{equation}
    Z_n  = \sqrt{I_X(\theta_0)n}(\widehat{\theta}_n-\theta_0) \; ,
\end{equation}
where $I_X(\theta_0)$ is the Fisher information defined as
\begin{align}
I_X(\theta_0) = -\mathbb{E}_{X|\theta_0}\left[ s_{\theta_0}'(X)\right]. 
\end{align}
and $(\theta,x ) \mapsto s_\theta(x)$ is the \emph{score function} given by 
\begin{equation}
    s_\theta (x) = \frac{\partial}{\partial \theta} \log f(x|\theta).
\end{equation}
The classical CLT for the maximum likelihood estimate can be stated as follows:
\begin{equation}
    Z_n \xrightarrow{d} Z
\end{equation}
as $n\to\infty$ where  $Z\sim\mathcal{N}(0,1)$.

We assume the Fisher information at this true parameter is nonzero. 
Finally, for simplicity we will assume that the model $f(x|\theta)$ is identifiable in the sense that $f(\cdot|\theta)=f(\cdot|\theta')$ only when $\theta=\theta'$.  The above assumptions are rather standard for the MLE setup.

\subsection{Additional Regularity Conditions}
\label{subsec:additional-regularity}
 We now collect the additional assumptions used to upgrade the classical distributional CLT to convergence in relative entropy. Assumption~\ref{ass:lipschitz-loglik} gives the uniform concentration needed for exponential consistency. Assumptions \ref{ass:lipschitz-loglik} and \ref{ass:lipschitz-score-derivative} control the second-order Taylor expansion of the likelihood and yield the high-moment bounds. Assumption~\ref{ass:bounded-fisher-information} is used only in the final step, where we pass from a smoothed estimator $\tilde{Z}_n$ back to $Z_n$.

All probabilistic conditions below are imposed under the true law $P_{\theta_0}$.

\begin{assumption}[Sub-Gaussian Lipschitz envelope for the log-likelihood]
\label{ass:lipschitz-loglik}
There exists a nonnegative function  $H: \cX \to \bbR^{+}$ such that, for all $\theta,\theta'\in K$ and all $x\in\mathcal X$,
\begin{equation}
    \left|\log f(x|\theta)-\log f(x|\theta')\right|
    \le H(x)|\theta-\theta'|.
    \label{eq:lipschitz-loglik}
\end{equation}
Moreover, $H(X)$ is sub-Gaussian under $P_{\theta_0}$.
\end{assumption}

\begin{rem}
By the mean-value theorem, one may take
\begin{equation}
    H(x)\le \sup_{u\in\Theta}|s_u(x)|
    \label{eq:H-score-bound}
\end{equation}
whenever the right-hand side is finite. Additionally, note that Assumption~1 implies that $s_\theta(X)$ is sub-Gaussian. 
\end{rem}

\begin{assumption}[Sub-exponential control of the score derivative]
\label{ass:lipschitz-score-derivative}
There exists a nonnegative function  $\bar H(X)$ such that, for all $\theta\in K$ and  all $x\in\mathcal X$,
\begin{equation}
    \left|s'_{\theta_0}(x)-s'_\theta(x)\right|
    \le \bar H(x)|\theta-\theta_0|.
    \label{eq:lipschitz-score-derivative}
\end{equation}
Moreover, both $s'_{\theta_0}(X)$ and $\bar H(X)$ are sub-exponential under $P_{\theta_0}$.
\end{assumption}

\begin{assumption}[Smoothness for the normalized estimator]
\label{ass:bounded-fisher-information}
Let $p_n$ denote the density of $Z_n$. Either
\begin{enumerate}
    \item there exist constants $N_0<\infty$ and $I_0<\infty$ such that
\begin{equation}
    \sup_{n\ge N_0}\int\frac{(p_n'(z))^2}{p_n(z)}\,dz\le I_0 \; ,
    \label{eq:fisher-density-bound}
\end{equation}
    or,
    \item there exist constants $N_0$ and $R$ such that $|p'_n(z)| \leq R$ uniformly in $z$ and $n\geq N_0$.
\end{enumerate}
\end{assumption}

\subsection{Remarks on Regularity Conditions}
The MLE representation in \eqref{eq:mleform} is not, in general, a sum of independent or weakly dependent random variables. Thus, the usual convolution and entropy-monotonicity methods for entropic CLTs do not apply directly. Instead, we use a method-of-moments approach. Assumptions~\ref{ass:lipschitz-loglik} and~\ref{ass:lipschitz-score-derivative} yield exponential consistency and uniform moment bounds for $Z_n$. More specifically, using Assumptions~\ref{ass:lipschitz-loglik} and~\ref{ass:lipschitz-score-derivative}, we can get convergence of $Z_n$ to $Z$ in terms  moments.
These two assumptions are also sufficient to prove the entropic CLT for the smoothed estimator
\begin{align}
\widetilde{Z}_n = \sqrt{1-\epsilon}\,Z_n + \sqrt{\epsilon}\,Z .
\end{align} 

We note that, without the additional Assumption \ref{ass:bounded-fisher-information}, one cannot in general upgrade the convergence of $Z_n$ to convergence in relative entropy. The reason is that the MLE is defined over the compact set $K$ in \eqref{eq:def_ml_estimation}, and may place positive probability on the boundary of $K$. Consequently, the distribution of $Z_n$ may contain both continuous and discrete components. In this case, $Z_n$ is not absolutely continuous with respect to the limiting Gaussian random variable $Z$, and hence $D(Z_n\|Z)=\infty$.

Therefore, entropic normality of $Z_n$ itself requires an additional regularity condition. Assumption~\ref{ass:bounded-fisher-information} is introduced precisely to rule out such singular behavior, for example discrete atoms in the distribution of $Z_n$. While this assumption may be difficult to verify when the MLE has no closed-form expression, it is not needed if one only seeks entropic normality of the smoothed  $\widetilde Z_n$.

\subsection{Examples of Maximum Likelihood}
If $f(\cdot| \theta)$ belongs to an exponential family, then the MLE is the sample mean of the sufficient statistic \cite{barndorff2014information}. Hence, the entropic CLTs of \cite{linnik1959information,barron1986entropy} already cover this case. Relatedly, the Pitman--Koopman--Darmois theorem \cite{lehmann1998theory} roughly says that exponential families are the only distributions for which the MLE is a function of a sample mean.

We now list a few examples beyond exponential family that will satisfy our regularity conditions:

\begin{itemize}
\item \emph{Pearson Type IV Location Family:}  For this family the parent distribution is given by 
\begin{equation} f(x |\theta)
\propto
\left(1+\frac{(x-\theta)^2}{\sigma^2}\right)^{-m}
{\rme}^{
\nu \arctan\frac{x-\theta}{\sigma}
}
\end{equation}
where $\sigma>0$, $m>\tfrac12$, and $\nu\in\mathbb{R}$ are fixed, and $\theta$ is the unknown location parameter \cite{abramowitz1970handbook}.  This family includes a large number of distribution such as Cauchy ($\nu =0, m=1)$ and  student-$t$ distribution ($\nu=0, m =\frac{\kappa+1}{2}$ where $\kappa$ is degree of freedom). The consistency and asymptotic normality of this family of distributions were considered in \cite{okamura2025asymptotics}; see also \cite{bai1987maximum} for the study of the MLE of the Cauchy distribution.   The score  and its derivative for this family is given by 
\begin{align}
s_\theta(x)&=
\frac{2m ( x-\theta)}{\sigma^2+( x-\theta)^2}-\frac{\nu\sigma}{\sigma^2+( x-\theta)^2}, \\
s'_\theta(x)
&=
\frac{2m(( x-\theta)^2-\sigma^2) - 2\nu\sigma ( x-\theta)}
{(\sigma^2+( x-\theta)^2)^2},
\end{align}
and  are both bounded functions. In particular, as shown in Appendix~\ref{app:pearson_lipt_sub_gaussianity}, the sub-Gaussianity of the Lipschtiz constant follows since  
\begin{equation}
    H(X) \le  \frac{m+|\nu|}{\sigma}
\end{equation}
Finally, since the score $s_\theta(x)$ is a bounded rational function of $x-\theta$ whose numerator and denominator are polynomials of finite degree, its second derivative with respect to $\theta$ is also uniformly bounded in $x$.
\item \emph{Logistic:} 
\begin{equation}
f(x  |\theta)=\frac{e^{-(x-\theta)}}{\bigl(1+e^{-(x-\theta)}\bigr)^2},\qquad x\in\mathbb{R}.
\end{equation}
The score and its derivatives are
\begin{align}
s_\theta(x)&=
\tanh\left(\frac{x-\theta}{2}\right), \\
s_\theta'(x)
&=-\frac12 \mathrm{sech}^2\left(\frac{x-\theta}{2}\right),
\end{align}
where both the score and its derivative are bounded. In particular, as shown in  Appendix~\ref{app:liptschits_for_logistic},  the sub-Gaussianity of the Lipschtiz constant follows since  
\begin{equation}
    H(X) \le  1.
\end{equation}
Similarly, it can be shown that the second derivative of the score function is also bounded. 
\item \emph{Cauchy Model with Unknown Scale:}
\begin{equation}
    f(x | \theta)=\frac{1}{\pi\theta}\frac{1}{1+(x/\theta)^2},\qquad x\in\mathbb{R}.
\end{equation} 
The score and its derivative are given by 
\begin{align}
    s_\theta(x) &=
\frac{x^2-\theta^2}{\theta(\theta^2+x^2)} \\
s_\theta'(x) & = -\frac{1}{\theta^2}
+\frac{2(\theta^2-x^2)}{(\theta^2+x^2)^2}
\end{align}
As shown in Appendix~\ref{app:example:Cauchy},  the sub-Gaussianity of the Lipschitz constant follows since  
\begin{equation}
    H(X) \le   \frac{1}{\theta_{\min}},
\end{equation}
where $\theta_{\min}>0$ is the minimum value of $\theta$ on $K$, which we assume to be compact. A similar bound applies to the second derivative, the score function.

This Cauchy example can be generalized to the Pearson VII scale family of distributions where the score function and its derivative would also be bounded.  
\end{itemize}

\section{Main Result}
\label{sec:main}

We now state the main results. We first establish moment control for $Z_n$, then prove entropic convergence for a smoothed version of $Z_n$, and finally remove the smoothing using Assumption~\ref{ass:bounded-fisher-information}.

\begin{thm}
\label{thm:sub_gauss}
Suppose that Assumptions~\ref{ass:lipschitz-loglik} and
\ref{ass:lipschitz-score-derivative} hold. Then
\begin{itemize}
    \item $Z_n$ is $O(1)$ sub-Gaussian; and
    \item $\mathbb{E}|Z_n|^m \to \mathbb{E}|Z|^m$ for every $m\geq 1$.
\end{itemize}
\end{thm}

\begin{IEEEproof}
See Section~\ref{sec:proof_theorem_sub_gauss}.
\end{IEEEproof}

\begin{thm}
\label{thm1}
 For a fixed $\epsilon\in(0,1)$, define
\begin{equation}
       \widetilde{Z}_n=\sqrt{1-\epsilon}\,Z_n+\sqrt{\epsilon}\,Z .
\end{equation}
Suppose that Assumptions~\ref{ass:lipschitz-loglik} and
\ref{ass:lipschitz-score-derivative} hold. Then, for every fixed $\epsilon\in(0,1)$,
\begin{equation}
     \lim_{n\to\infty}D(\widetilde Z_n\|Z)=0 .
\end{equation}

\end{thm}
\begin{IEEEproof}
See Section~\ref{sec:KL}.
\end{IEEEproof}

 The next proposition shows
that, if the densities of the original variables $Z_n$ have uniformly bounded
Fisher information, then the smoothing can be removed.

\begin{cor}
\label{prop:without_smooth}
Suppose that Assumptions~\ref{ass:lipschitz-loglik},
\ref{ass:lipschitz-score-derivative}, and
\ref{ass:bounded-fisher-information} hold. Then, for $Z\sim\mathcal N(0,1)$,
\begin{equation}
     \lim_{n\to\infty}D(Z_n\|Z)=0 .
\end{equation}
   
\end{cor}
\begin{IEEEproof}
See Section~\ref{proof:smooth}.
\end{IEEEproof}

\section{Proofs}
\label{sec:proofs}

\subsection{Proof of Theorem~\ref{thm:sub_gauss}}
\label{sec:proof_theorem_sub_gauss}
\subsubsection{Consistency and a Uniform Law of Large Numbers} \label{subsec1}

Let $L(\theta) = \mathbb{E}_{\theta_0}[\log f(X|\theta)]$ which is maximized at $\theta=\theta_0$ by the non-negativity of KL. This maximum is unique assuming (as we do) that the model is identifiable. We now consider a uniform law of large large numbers result to compare $L(\theta)$ and $L_n(\theta)$, with explicit bounds on the probability of deviation.

By the Lipschitz assumption,
\begin{align}
|L_n(\theta) - L_n(\theta_0)| & \leq |\theta-\theta_0|\frac{1}{n}\sum_{i=1}^n H(X_i) \label{eq:log_lik_lip}
\end{align}
where $\frac{1}{n}\sum_{i=1}^n H(X_i)$ is sub-Gaussian with parameter $O(1/n)$ (when its mean is removed; it has an $O(1)$ mean).

For a small $\varepsilon>0$, consider the set $K' = K\setminus \mathsf{Ball}(\theta_0,\varepsilon)$. Let $D=\max_{\theta\in K'}L(\theta)$ and note $\Delta = L(\theta_0) - D > 0$ by the continuity of $L$ and identifiability. Again, by the continuity of $L_n$ on compact $K'$, there exists a $\delta>0$ such that $|\theta-\theta'|<\delta \implies |L_n(\theta)-L_n(\theta')|<\Delta/2$ for $\theta,\theta'\in K'$. Importantly, we can pick $\delta$ to be independent of $n$ and the particular samples $X_1,\ldots,X_n$ with high probability (i.e., at least probability $1-e^{-nc}$ for some $c>0$) by using a sub-Gaussian tail bound with \eqref{eq:log_lik_lip}, so we will assume this is the case without loss of generality. Consider the open cover $\{\mathsf{Ball}(\theta,\delta) \; : \; \theta\in K'\}$ of $K'$. By compactness there exists a finite subcover with corresponding ball centers $\theta_1,\ldots,\theta_N$.

By the fundamental theorem of calculus,
\begin{align}
L_n(\theta) - L_n(\theta_0) = \frac{1}{n}\sum_{i=1}^n \int^\theta_{\theta_0} \frac{\partial}{\partial t} \log f(X_i|t) dt
\end{align}
and therefore
\begin{align} \label{eq:difference_deviation}
& (L_n(\theta) - L(\theta)) - (L_n(\theta_0)-L(\theta_0)) \\
& = \frac{1}{n}\sum_{i=1}^n \int^\theta_{\theta_0} \frac{\partial}{\partial t} \log f(X_i|t) dt + L(\theta_0) - L(\theta) \\
& = \frac{1}{n}\sum_{i=1}^n \int^\theta_{\theta_0} \frac{\partial}{\partial t} \log f(X_i|t) dt - \mathbb{E}_{\theta_0}\left[\int^\theta_{\theta_0}\frac{\partial}{\partial t}\log f(X|t)dt\right]. \label{eq:difference_deviation_2}
\end{align}
Zooming in on a single term from the right-hand side of \eqref{eq:difference_deviation_2}, $\int^\theta_{\theta_0} \frac{\partial}{\partial t} \log f(X_i|t) dt$, where $\frac{\partial}{\partial t} \log f(X_i|t)$ has mean $\mu_t$, we have the centralized moments
\begin{align}
& \mathbb{E}_{\theta_0} \left| \int^\theta_{\theta_0} \frac{\partial}{\partial t} \log f(X_i|t) dt - \mathbb{E}_{\theta_0}\left[\int^\theta_{\theta_0}\frac{\partial}{\partial t}\log f(X|t)dt\right] \right|^p \\
& = \mathbb{E}_{\theta_0} \left| \int^\theta_{\theta_0} \left[ \frac{\partial}{\partial t} \log f(X_i|t) dt - \mathbb{E}_{\theta_0}\left[\frac{\partial}{\partial t}\log f(X|t)\right]\right]dt \right|^p \\
& \leq |\theta-\theta_0|^{p-1} \mathbb{E}_{\theta_0} \int^\theta_{\theta_0}\left|  \frac{\partial}{\partial t} \log f(X_i|t) - \mu_t \right|^p dt \\
& = |\theta-\theta_0|^{p-1}  \int^\theta_{\theta_0} \mathbb{E}_{\theta_0} \left|  \frac{\partial}{\partial t} \log f(X_i|t) - \mu_t\right|^p dt \\
& \leq |\theta-\theta_0|^p(\sigma \sqrt{p})^p \\
& \leq (2B\sigma\sqrt{p})^p \; . \label{eq:subgauss_const}
\end{align}
From \eqref{eq:subgauss_const} it follows that \eqref{eq:difference_deviation} is sub-Gaussian with parameter $O(1/n)$.

Using a sub-Gaussian tail bound and a union bound over $\theta_1,\ldots,\theta_N$, we have
\begin{align} \label{eq:event}
\mathsf{Pr}\bigg( \big|(L_n(\theta_i) - L(\theta_i)) - (L_n(\theta_0) - L(\theta_0))\big|  > \Delta/2  \notag\\
\; \text{ for some } i=1,\ldots,N \bigg)  < 2N\exp(-cn\Delta^2) \; .
\end{align}
For any $\theta \in K'$, there exists a $\theta_i$ for some $i=1,\ldots,N$ with $|\theta - \theta_i|<\delta$ and therefore $|L_n(\theta) - L_n(\theta_i)| < \Delta/2$. Suppose we are in the complement of the event in \eqref{eq:event}. Then 
\begin{align}
    \big|(L_n(\theta_i) - L(\theta_i)) - (L_n(\theta_0) - L(\theta_0))\big| \leq \Delta/2
\end{align}
for all $i=1,\ldots,n$, and
\begin{align}
&(L_n(\theta) - L_n(\theta_i)) + (L_n(\theta_i) - L(\theta_i))  \notag\\
&- (L_n(\theta_0) - L(\theta_0))  \\
&< \Delta \\
& \leq L(\theta_0) - L(\theta_i)
\end{align}
so that $L_n(\theta) - L_n(\theta_0) < 0$ (for arbitrary $\theta\in K'$) and $\widehat\theta_n$ cannot be in $K'$. Thus 
\begin{align} \label{eq:ulln}
|\widehat{\theta}_n - \theta_0| < \varepsilon \text{ with prob. at least }1-2N\exp(-cn\Delta^2) \; .
\end{align}
The key thing that differentiates this consistency from the usual results is the exponentially decreasing probability of deviation in \eqref{eq:ulln}. This will be needed to get convergence in moments (for any arbitrarily large moment) below.

\subsubsection{Convergence in Distribution}
By the MVT we have
\begin{equation}
    L_n'(\widehat{\theta}_n) =  L_n'(\theta_0)  +  L_n''(\theta^c_n)  ( \widehat{\theta}_n - \theta_0)
\end{equation}
where the random variable $\theta^c_n $ lies between $\widehat{\theta}_n$ and $\theta_0$. Using consistency $L'_n(\widehat\theta_n) = 0$ with probability approaching one, so we are able to ignore this term.
Now note that
\begin{equation} \label{eq:mleform}
    \sqrt{n}  ( \widehat{\theta}_n - \theta_0)  = -\frac{ \sqrt{n}L_n'(\theta_0)}{L_n''(\theta^c_n)} \; .
\end{equation}
The numerator
\begin{equation}
 \sqrt{n}L_n'(\theta_0) =   \frac{1}{\sqrt{n}}\sum_{i=1}^n  s_{\theta_0}(X_i)
\end{equation}
converges in distribution to normal with zero mean and variance given by $I_X(\theta_0) = \mathsf{var}(  s_{\theta_0}(X) ) $. 
The denominator can be controlled as follows:
\begin{align}
\big| L_n''(\theta^c_n) - L_n''(\theta_0) \big| & =  \left|\frac{1}{n}\sum_{i=1}^n s_{\theta_n^c}'(X_i) - \frac{1}{n}\sum_{i=1}^n  s_{\theta_0}'(X_i) \right| \\
& \leq \big|\theta_n^c - \theta_0\big|\frac{1}{n}\sum_{i=1}^n \bar{H}(X_i) \; , \label{eq:Lipshitz_sec_der_log_like}
\end{align}
where the last bound follows from the assumption that  $ s_{\theta}'(x)$ has the Lipschitz constant $\bar{H}(x)$. 
From \eqref{eq:Lipshitz_sec_der_log_like}  and by the law of large numbers and consistency, we have that  $L_n''(\theta^c_n)$ converges in probability to the same value that $L_n''(\theta_0)$ does, which is $\bbE\left[  s_{\theta_0}'(X_i) \right] = -I_X(\theta_0)$. Note that from the differentiability assumptions we have made the two definitions of the Fisher information coincide
\begin{equation}
- \bbE\left[  s_{\theta_0}'(X_i) \right] = \mathsf{var}(  s_{\theta_0}(X) ). 
\end{equation}
  The proof of convergence in distribution is concluded by using a standard argument via  Slutsky's theorem \cite{billingsley1999convergence} this implies $\sqrt{n}  ( \widehat{\theta}_n - \theta_0)$ converges in distribution to $\mathcal{N}(0,I_X(\theta_0)^{-1})$. 

\subsubsection{Convergence in Moments}
Consider the $m$th moment of $\sqrt{n}  ( \widehat{\theta}_n - \theta_0)$,
\begin{align}
\mathbb{E} \left[|\sqrt{n}  ( \widehat{\theta}_n - \theta_0)|^m \right] & = \int_{E_1} |\sqrt{n}  ( \widehat{\theta}_n - \theta_0)|^m \rmd P_{\theta_0}^n  \notag\\
&+ \int_{E_2} \left| \frac{ \sqrt{n}L_n'(\theta_0)}{L_n''(\theta^c_n)}\right|^m  \rmd P_{\theta_0}^n
\end{align}
where we have split the integral into a sum of integrals over disjoint events $E_1$ and $E_2$. We take $E_1$ to be the ``bad'' set, that is, the set of outcomes $X_1,\ldots,X_n$ such that $|\widehat{\theta}_n-\theta_0|\geq \varepsilon$ or $|-L''_n(\theta^c_n) - I_X(\theta_0)| \geq 3\mathbb{E}[\bar{H}(X)]\varepsilon$. Following from \eqref{eq:ulln}, the first of these conditions happens with probability that is exponentially decaying in $n$. If this is not the case (i.e., $|\widehat{\theta}_n-\theta_0| < \varepsilon$), then
\begin{align}
&|-L''_n(\theta^c_n) - I_X(\theta_0)| \notag\\
& \leq |L''_n(\theta^c_n) - L''_n(\theta_0)| + |-L''_n(\theta_0) - I_X(\theta_0)| \\
&  \leq \frac{\varepsilon}{n}\sum_{i=1}^n \bar{H}(X_i) + |-L''_n(\theta_0) - I_X(\theta_0)| \label{eq:triang_bound}
\end{align}
where each term in \eqref{eq:triang_bound} is less than $3\mathbb{E}[\bar{H}(X)]\varepsilon/2$ with probability decaying exponentially in $n$ (using the sub-exponential tail bounds assumed in the theorem).

Thus $\mathsf{Pr}(E_1)$ also decays with probability exponentially decreasing in $n$ and
\begin{equation}
    \int_{E_1} |\sqrt{n}  ( \widehat{\theta}_n - \theta_0)|^m \rmd P_{\theta_0}^n  \leq (2B\sqrt{n})^m\mathsf{Pr}(E_1) \to 0
\end{equation}
as $n\to\infty$.

The sub-Gaussianity of the Lipschitz constant $H(X)$ immediately implies a similar moment bound on the score, i.e.,
\begin{equation}
\int_{E_2} \left| \sqrt{n}L_n'(\theta_0)\right|^m \rmd P_{\theta_0}^n \leq (\sigma \sqrt{m})^m \; .
\end{equation}
As $\varepsilon\to 0$, the denominator $L''_n(\theta^c_n)$ becomes close to $I_X(\theta_0)$ in the event $E_2$, so for small enough $\varepsilon$ the integral over this event satisfies
\begin{equation}
    \int_{E_2} \left| \frac{ \sqrt{n}L_n'(\theta_0)}{L_n''(\theta^c_n)}\right|^m  \rmd P_{\theta_0}^n  \leq \left(\frac{2}{I_X(\theta_0)}\sigma \sqrt{m}\right)^m \; .
\end{equation}

Noting that
\begin{equation}
 \sup_n \; (2B\sqrt{n})^me^{-cn} = O(m^\frac{m}{2}) \; ,
\end{equation}
we see the whole integral satisfies
\begin{equation} \label{eq:moment_bound}
\mathbb{E}|\sqrt{n}  ( \widehat{\theta}_n - \theta_0)|^m \leq (C\sqrt{m})^m
\end{equation}
for sufficiently large $C$, and therefore $Z_n$ itself is also sub-Gaussian. 

Using uniform integrability, the moment bound in \eqref{eq:moment_bound} along with the convergence in distribution implies convergence in moments, i.e., $\mathbb{E}|Z_n|^m \to \mathbb{E}|Z|^m$ for all $m$.

\subsection{Proof of Theorem~\ref{thm1}}
\label{sec:KL}

We have the following existing bound on the KL divergence from a standard normal density $\phi$; see  \cite[Lemma 2.2]{bobkov2014berry}.
\begin{prop} \label{prop1}
For all $T\geq 0$,
\begin{align}
    &D(Z_n \|Z) \leq  \exp\left(-\frac{T^2}{2}\right)  \notag\\
    &+ \sqrt{2\pi}\int_{-T}^T (p_n(z) - \phi(z))^2 \exp\left(\frac{z^2}{2}\right) \rmd z \label{eq:bound_term1}\\
    & + \frac{1}{2}\int_{|z|\geq T} z^2p_n(z) \rmd z + \int_{|z|\geq T} p_n(z)\log p_n(z)  \rmd z \; . \label{eq:bound_term2}
\end{align}
\end{prop}
The first term in \eqref{eq:bound_term1} goes to zero as $T$ gets large, so let us focus on the second term.
\begin{align}
   & \int_{-T}^T (p_n(z) - \phi(z))^2 \exp\left(\frac{z^2}{2}\right) \rmd z \notag\\
    & \leq \exp\left(\frac{T^2}{2}\right) \int_{-T}^T (p_n(z) - \phi(z))^2 \rmd dz \\
     \leq & \exp\left(\frac{T^2}{2}\right) \int_{-\infty}^\infty (p_n(z) - \phi(z))^2 \rmd z \\
     = & \exp\left(\frac{T^2}{2}\right) \int_{-\infty}^\infty \big|\widehat{p}_n(\omega) - \widehat\phi(\omega)\big|^2 \rmd \omega \label{eq:using_parsevals} \\
     = & \exp\left(\frac{T^2}{2}\right) \int_{-M}^M \big|\widehat{p}_n(\omega) - \widehat\phi(\omega)\big|^2 \rmd \omega \label{eq:char_func_conv} \\
    & + 2\exp\left(\frac{T^2}{2}\right) \int_{M}^\infty \big|\widehat{p}_n(\omega) - \widehat\phi(\omega)\big|^2 \rmd \omega \label{eq:char_func_tail}
\end{align}
where \eqref{eq:using_parsevals} is due to Parseval's theorem where we use the following convention for the characteristic function: $\widehat{p}_n(\omega)  = \int p_n(z)e^{-2\pi i z \omega} \rmd z$. 
In the next subsection, we will show that the tail of the characteristic function difference \eqref{eq:char_func_tail} goes to zero as $M\to\infty$ uniformly in $n$. For the term \eqref{eq:char_func_conv}, we use the Taylor expansion
\begin{align}
\widehat{p}_n(\omega) & = \int p_n(z)e^{-2\pi i z \omega} \rmd z \\
& = \int p_n(z) \left(\sum_{m=0}^\infty \frac{(-2\pi i z)^m}{m!}\omega^m\right) \rmd z \\
& = \sum_{m=0}^\infty \frac{(-2\pi i)^m}{m!}\omega^m\left(\int z^m p_n(z) \rmd z\right) \label{eq:taylor_exp} \; .
\end{align}
Because the estimates $Z_n$ themselves are sub-Gaussian random variables,
\begin{equation}
\omega^m \mathbb{E} \left[ |Z_n|^m \right] \leq (M C \sqrt{m})^m \label{eq:mom_bound}
\end{equation}
for all $|\omega|\leq M$. Thus, by Stirling's approximation \eqref{eq:taylor_exp} converges uniformly on the set $|\omega|\leq M$ (and also the interchange of limits in \eqref{eq:taylor_exp} is justified).
Using both \eqref{eq:mom_bound} (and Stirling's approximation) to dominate convergence, and the convergence in moments that was shown in the last subsection, we get
\begin{equation}
\lim_{n\to\infty} \int_{-M}^M \big|\widehat{p}_n(\omega) - \widehat\phi(\omega)\big|^2 \rmd\omega = 0 \; .
\end{equation}

The first integral in \eqref{eq:bound_term2} can be bounded by the sub-Gaussianity of $Z_n$ and the Cauchy-Schwarz inequality,
\begin{align}
&\int 1_{\{|z|\geq T\}}(z) z^2p_n(z) \rmd z \notag\\
& \leq \left( \int 1_{\{|z| \geq T\}}(z) p_n(z) \rmd x\int z^4p_n(z) \rmd z \right)^\frac{1}{2} \; ,
\end{align}
which goes to  as $T\to\infty$ since $ \int z^4p_n(z) \rmd z<\infty$ \; .

There are two terms that remain to be controlled, 
\begin{equation} \label{eq:ent_tail}
    \int_{|z|\geq T} p_n(z)\log p_n(z) \rmd z
\end{equation}
from \eqref{eq:bound_term2}, and 
\begin{equation} \label{eq:fourier_tail}
    \int_{M}^\infty \big|\widehat{p}_n(\omega) - \widehat\phi(\omega)\big|^2 \rmd\omega
\end{equation}
from \eqref{eq:char_func_tail}. To deal with these terms, we will consider the smoothed version of $Z_n$ that is defined by
\begin{equation}
\widetilde{Z}_n = \sqrt{1-\epsilon}\,Z_n + \sqrt{\epsilon}\,Z \; .
\end{equation}
Let $\widetilde{p}_n$ be the corresponding density for $\widetilde{Z}_n$. By taking the sum of independent sub-Gaussian random variables, all of the above moment bounds apply to $\widetilde{Z}_n$ just as they did to $Z_n$. Also, by the continuous mapping theorem, we still have $\widetilde{Z}_n\xrightarrow{d} Z$.

For \eqref{eq:ent_tail},
\begin{align}
\sup_z \widetilde{p}_n(z) & \leq \frac{1}{\sqrt{\epsilon(1-\epsilon)}} \int p_n\left(\frac{z-x}{\sqrt{1-\epsilon}}\right)\phi\left(\frac{x}{\sqrt{\epsilon}}\right) \rmd x \\
& \leq \frac{1}{\sqrt{2\pi\epsilon}}
\end{align}
so that
\begin{equation}
    \int_{|z|\geq T} \hspace{-0.1cm}\widetilde{p}_n(z)\log \widetilde{p}_n(z) \rmd z \leq \left(\log\frac{1}{\sqrt{2\pi\epsilon}} \right)\int_{|z|\geq T} \hspace{-0.1cm}\widetilde{p}_n(z) \rmd z
\end{equation}
which goes to zero as $T\to\infty$.

For the tail of the characteristic function,
\begin{align}
    &\int_{M}^\infty \big|\widehat{\widetilde{p}}_n(\omega) - \widehat\phi(\omega)\big|^2 \rmd \omega  \notag \\
    & \leq 2\int_{M}^\infty \big|\widehat{\widetilde{p}}_n(\omega)\big|^2 \rmd \omega + 2 \int_{M}^\infty \big|\widehat\phi(\omega)\big|^2 \rmd \omega \\
    & \leq 4\int_M^\infty e^{-c\omega^2} \rmd \omega, \label{Eq:Bound_on_diff}
\end{align}
where the last inequality follows since the characteristic function of the sum  can be bounded as $\big|\widehat{\widetilde{p}}_n(\omega)\big| = \big|\widehat{p}_n( \sqrt{1-\epsilon}\omega) \hat{\phi} (\sqrt{\epsilon} \omega) \big| \le  \big | \hat{\phi} (\sqrt{\epsilon} \omega) \big| $. We note that \eqref{Eq:Bound_on_diff} 
 goes to zero as $M\to\infty$. 

Taking $n\to\infty$, then $M\to\infty$, and finally $T\to\infty$, yields
\begin{align}
       \lim_{n \to \infty} D(\widetilde{Z}_n \| Z) = 0 \; .
\end{align}

\subsection{Proof of Corollary~\ref{prop:without_smooth}}
\label{proof:smooth}

The proof follows by taking $\epsilon \to 0$ and use the final regularity condition, i.e. the boundedness of the Fisher information for the density $p_n(z)$, to ensure that this changes the resulting relative entropy $D(Z_n\|Z)$ in a continuous way. To this end, we use the following (slightly modified version of Lemma 1) from \cite{barron1986entropy}.
\begin{prop}
If $\sigma_n^2$ is the variance of $Z_n$ and
\begin{equation}J(Z_n) = \sigma_n^2 \mathbb{E}\left[\left(\frac{p'_n(Z_n)}{p_n(Z_n)} + \frac{Z_n}{\sigma_n^2}\right)^2\right]
\end{equation}
is the standardized Fisher information for $Z_n$, then
\begin{align}
D(Z_n\|Z)-D(\widetilde{Z}_n\|Z) & = \int_0^\epsilon \frac{  J(\sqrt{1-t}Z_n + \sqrt{t}Z) \rmd t}{2(1-t)} \\
& \leq \int_0^\epsilon J(Z_n)\frac{ \rmd t}{2} = \frac{\epsilon J(Z_n)}{2} \; .
\end{align}
\end{prop}

Furthermore,
\begin{equation}
J(Z_n) = -1 + \sigma_n^2 I(Z_n) \leq -1 + \sigma_n^2 I_0
\end{equation}
where $I(Z_n)$ is the Fisher information associated with the density $p_n$. The variance $\sigma_n^2 \to 1$ by convergence in moments, and thus
\begin{equation}
D(Z_n\|Z)-D(\widetilde{Z}_n\|Z) \leq \frac{\epsilon I_0}{2}
\end{equation}
for sufficiently large $n$. Therefore $\limsup_{n\to\infty} D(Z_n\|Z) \leq \epsilon I_0/2$ for any $\epsilon>0$ and taking $\epsilon \to 0$ we arrive at $D(Z_n\| Z) \to 0$ as desired.

 If instead we have the alternate smoothness condition that 
\begin{align}|p'_n(z)|<R \;,\label{eq:deriv_bound}\end{align} we note that both terms \eqref{eq:ent_tail} and \eqref{eq:fourier_tail} can be bounded by terms that go zero as $T,M\to\infty$. For \eqref{eq:ent_tail} note that \eqref{eq:deriv_bound} implies
\begin{align} \label{eq:property2}
\sup_z p_n(z) \leq R
\end{align}
since $p_n(z)$ is a probability density. For \eqref{eq:fourier_tail} we note that
\begin{align} \label{eq:property1}
\int_M^\infty |\widehat{p}_n(\omega)|^2 d\omega = O\left(\frac{1}{M}\right) \; .
\end{align}
This can be seen as follows. By the Mean Value Theorem and the non-negativity of $_n$, we can bound the squared difference for any shift $h > 0$:
\begin{align}
(p_n(z+h) - p(z))^2 \le Rh(p_n(z+h) + p_n(z))
\end{align}
Because $p_n(z)$ integrates to 1, integrating both sides over $\mathbb{R}$ gives the $L^2$ translation bound:
\begin{align} 
\int_{-\infty}^{\infty} & (p_n(z+h) - p_n(z))^2 \, dz \\
&\le Rh \int_{-\infty}^{\infty} (p_n(z+h) + p_n(z)) \, dx = 2Rh \; .\label{eq:modcont}
\end{align}
Using standard Fourier analysis, the $L^2$ modulus of continuity \eqref{eq:modcont} implies \eqref{eq:property1}.


\section{Conclusion} \label{sec:conclusion}
In this paper, we derived an entropic version of the Central Limit Theorem (CLT) for the maximum likelihood estimator (MLE), establishing a stronger mode of convergence than the classical distributional CLT. A natural direction for future work is to extend this line of reasoning beyond MLEs. One promising avenue---likely to benefit from similar technical tools---is the study of M-estimators, where the score function is replaced by a general estimating function with comparable smoothness and boundedness properties.

Another natural extension is the development of multivariate versions of the results presented here. Since the entropic and Fisher-information--based techniques used in this work are inherently vectorial, we expect that the core arguments should carry over with minimal conceptual modification.

A closely related line of research investigates the monotonicity properties of estimators. For example, it is known that the sample mean satisfies a monotonicity property: the mapping
\[
n \mapsto D\!\left(\frac{1}{\sqrt{n}} \sum_{i=1}^n X_i \,\bigg\|\, Z\right)
\]
is non-increasing~\cite{courtade2016monotonicity}. In contrast, central order statistics, while obeying the entropic CLT, do not obey a similar monotonicity property~\cite{cardone2022entropic}. It would therefore be interesting to investigate whether and under what conditions an analogous phenomenon can be established in the context of the MLE.

\bibliography{refs.bib}
\bibliographystyle{IEEEtran}

\appendices

\section{Lipschitz constant for Pearson Example}
\label{app:pearson_lipt_sub_gaussianity}

To bound the Lipschitz constant, we use the bound in \eqref{eq:H-score-bound}. Since the score is given by
\begin{equation}
    s_\theta(x)
=\frac{2m(x-\theta)}{\sigma^2+(x-\theta)^2}
-\frac{\nu\sigma}{\sigma^2+(x-\theta)^2}.
\end{equation}
Let $y=x-\theta$. Then
\[
|s_\theta(x)|
\le
\frac{2m|y|}{\sigma^2+y^2}
+
\frac{|\nu|\sigma}{\sigma^2+y^2}.
\]
The function $y\mapsto \frac{2m|y|}{\sigma^2+y^2}$ is maximized at $|y|=\sigma$
and equals $\frac{m}{\sigma}$, while
$y\mapsto \frac{|\nu|\sigma}{\sigma^2+y^2}$ is maximized at $y=0$
and equals $\frac{|\nu|}{\sigma}$.
Therefore, for all $\theta$ and $x$,
\begin{equation}
    |s_\theta(x)| \le \frac{m+|\nu|}{\sigma}.
\end{equation}

This bound is uniform over $\theta\in K$, proving the claim.

\section{Lipschitz constant for the logistic example}
\label{app:liptschits_for_logistic}

To bound the Lipschitz constant, we use the bound in \eqref{eq:H-score-bound}. Recall, that for logistic example the score is given by 
\begin{equation}
s_\theta(x)=\tanh\!\left(\frac{x-\theta}{2}\right).
\end{equation}
Since $|\tanh(u)|\le 1$ for all $u\in\mathbb{R}$, we have
\begin{equation}
   \sup_{u\in K}|s_u(x)| \le 1. 
\end{equation}
This proves the claim with $H \le 1$.

\section{Lipschitz constant for Cauchy Example}
\label{app:example:Cauchy}
The score function is given by 
\begin{equation}
    s_\theta(x)
=\frac{x^2-\theta^2}{\theta(\theta^2+x^2)}.
\end{equation}
Using $|x^2-\theta^2|\le x^2+\theta^2$,
\begin{equation}
|s_\theta(x)|
=\frac{|x^2-\theta^2|}{\theta(\theta^2+x^2)}
\le
\frac{x^2+\theta^2}{\theta(\theta^2+x^2)}
=\frac{1}{\theta}.
\end{equation}
Since $\theta\ge \theta_{\min}$ for all $\theta\in K$,
\begin{equation}
    |s_\theta(x)| \le \frac{1}{\theta_{\min}}.
\end{equation}
This bound is uniform in both $x$ and $\theta$, proving the claim.

\end{document}